\documentclass[aps, pra, amsmath, amssymb, amsfonts, twocolumn,superscriptaddress]{revtex4-2}
\usepackage{mathrsfs}
\usepackage{amsthm}
\usepackage{graphicx}
\usepackage{subfigure} 
\usepackage{algorithm}
\usepackage[noend]{algpseudocode}
\usepackage{float}
\usepackage{dcolumn}
\usepackage{bm}
\usepackage{bbold}
\usepackage{textcomp}
\usepackage{amsmath}
\usepackage[titletoc]{appendix}
\usepackage{tabularx}
\usepackage{verbatim}
\usepackage{booktabs}
\usepackage{soulutf8,color,xcolor}

\sethlcolor{white}
\soulregister\cite7
\soulregister\ref7

\usepackage{hyperref}
\hypersetup{                    
	colorlinks,
	citecolor=blue,
	filecolor=blue,
	linkcolor=blue,
	urlcolor=blue
}

\newtheorem{theorem}{Theorem}

\newcounter{RomanNumber}

\begin{document}

\title{Speeding up the classical simulation of Gaussian boson sampling with limited connectivity}

\author{Tian-Yu Yang}
\affiliation{ State Key Laboratory of Low Dimensional Quantum Physics, Department of Physics, \\
	Tsinghua University, Beijing 100084, China}
\author{Xiang-Bin Wang}
\email[Corresponding aurthor ]{xbwang@mail.tsinghua.edu.cn}
\affiliation{ State Key Laboratory of Low Dimensional Quantum Physics, Department of Physics, \\ Tsinghua University, Beijing 100084, China}
\affiliation{ Jinan Institute of Quantum technology, SAICT, Jinan 250101, China}
\affiliation{ Shanghai Branch, CAS Center for Excellence and Synergetic Innovation Center in Quantum Information and Quantum Physics, University of Science and Technology of China, Shanghai 201315, China}
\affiliation{ Shenzhen Institute for Quantum Science and Engineering, and Physics Department, Southern University of Science and Technology, Shenzhen 518055, China}
\affiliation{ Frontier Science Center for Quantum Information, Beijing 100084, China}
\centerline{}
\begin{abstract}
Gaussian Boson sampling (GBS) plays a crucially important role in demonstrating quantum advantage. 
As a major imperfection, the limited connectivity of the linear optical network weakens the quantum advantage result in recent experiments. Here we present a faster classical algorithm to simulate the GBS process with limited connectivity.
In this work, we introduce an enhanced classical algorithm for simulating GBS processes with limited connectivity.
It computes the loop Hafnian of an $n \times n$ symmetric matrix with bandwidth $w$ in $O(nw2^w)$ time which is better than the previous fastest algorithm which runs in $O(nw^2 2^w)$ time. This classical algorithm is helpful on clarifying how limited connectivity affects the computational complexity of GBS and tightening the boundary of quantum advantage in the GBS problem.

\end{abstract}

\maketitle

\section{Introduction}
Gaussian Boson sampling (GBS) is a variant of Boson sampling (BS) that was originally proposed to demonstrate the quantum advantage~\cite{aaronson_computational_2011,hamilton_gaussian_2017,kruse_detailed_2019,deshpande_quantum_2022}. In recent years, great progress has been made in experiments on GBS~\cite{zhong_quantum_2020, arrazola_quantum_2021, zhong_phase-programmable_2021, madsen_quantum_2022, thekkadath_experimental_2022, sempere-llagostera_experimentally_2022, deng_solving_2023}. Both the total number of optical modes and detected photons in GBS experiments have surpassed several hundred~\cite{madsen_quantum_2022,zhong_phase-programmable_2021}. Moreover, it is experimentally verified that GBS devices can enhance the classical stochastic algorithms in searching some graph features~\cite{sempere-llagostera_experimentally_2022,deng_solving_2023}.

The central issue in GBS experiments is to verify the quantum advantage of the result.
\hl{The time cost of the best known classical algorithm for simulating an ideal GBS process grows exponentially with the system size~\cite{bjorklund_faster_2019}.} Therefore a quantum advantage result might be achieved when the system size is large enough~\cite{bulmer_boundary_2022,quesada_quadratic_2022,brod_photonic_2019}.

However, there are always imperfections in real quantum setups, and hence the time cost of corresponding classical simulation will be reduced. 
When the quantum imperfection is too large, the corresponding classical simulation methods can work efficiently~\cite{qi_regimes_2020,qi_efficient_2022,oh_classical_2022}. \hl{In this situation, a quantum advantage result won't exist even if the system size of a GBS experiments is very large.} Therefore, finding better methods to classically simulate the imperfect GBS process is rather useful in exploring the tight criteria for quantum advantage of a GBS experiment.


A major imperfection of recent GBS experiments is the shallow optical circuit~\cite{oh_classical_2022,qi_efficient_2022,zhong_quantum_2020,zhong_phase-programmable_2021}. A shallow optical circuit has limited connectivity and its transform matrix will deviate from the global Haar-random unitary~\cite{russell_direct_2017,oh_classical_2022}. 
In the original GBS protocol~\cite{hamilton_gaussian_2017,kruse_detailed_2019}, the unitary transform matrix $U$ of the passive linear optical network should be randomly chosen from Haar measure. However, due to the photon loss which increases exponentially with the depth of the optical network, the optical network might not be deep enough to meet the requirements of the full connectivity and the global Haar-random unitary~\cite{russell_direct_2017}. 
Classical algorithms can take advantage of the bandwidth structure and the deviation from the global Haar-random unitary to realize a speed-up in simulating the whole sampling process~\cite{qi_efficient_2022,oh_classical_2022}. 
Actually, with limited connectivity in the quantum device, the speed-up of corresponding classical  simulation can be exponentially~\cite{oh_classical_2022}. 

The most time-consuming part of simulating the GBS process with limited connectivity is to calculate the loop Hafnian of banded matrices.  A classical algorithm to calculate the loop Hafnian of a banded $n\times n$ symmetric matrix with bandwidth $w$ in time $O\left(nw 4^w\right)$ is given in Ref.~\cite{qi_efficient_2022}. Later, an algorithm that  takes time $O\left(nw^2 2^w\right)$ is given in Ref.~\cite{oh_classical_2022}. Here we present a classical algorithm to calculate the loop Hafnian of a banded $n\times n$ matrix with bandwidth $w$ in time $O\left(nw 2^w\right)$. We also show that this algorithm can be used to calculate the loop Hafnian of sparse matrices.

Our algorithm reduces the time needed for classically simulating the GBS process with limited connectivity. This is helpful in clarifying how limited connectivity affects the computational complexity of GBS and tightening the boundary of quantum advantage in the GBS problem.

This article is organized as follows. In Sec.~\ref{sec:2}, we give an overview of the background knowledge which will be used later. In Sec.~\ref{sec:3}, we give our improved classical algorithm for calculating the loop Hafnian of banded matrices.
Finally, we make a summary  in Sec.~\ref{sec:4}.

\section{Overview of Gaussian Boson sampling and limited connectivity}\label{sec:2}
\subsection{Gaussian Boson sampling protocol} \label{sec:2_1}
In the GBS protocol, $K$ single-mode squeezed states (SMSSs) are injected into an $M$-mode passive linear optical network and detected in each output mode by a photon number resolving detector. The detected photon number of each photon number resolving detector forms an output sample which can be denoted as $\bar{n}=n_1 n_2 ...n_M$. The schematic setup of Gaussian boson sampling is shown in Fig.~\ref{fig:1}

\begin{figure}
	\centering
	\includegraphics[width=0.8\linewidth]{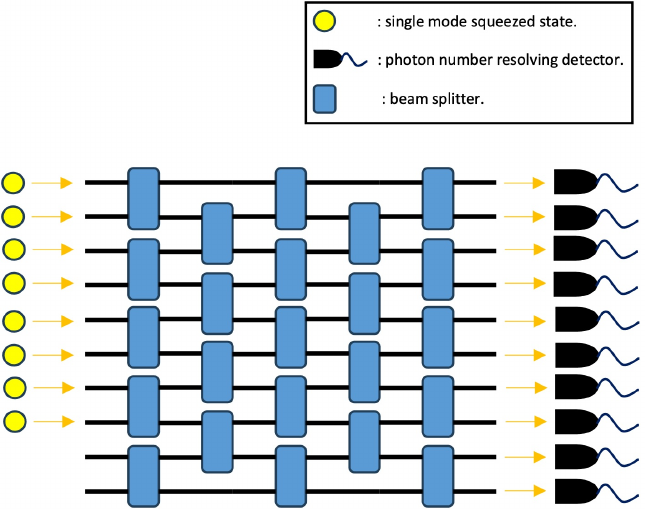}
	\caption{A schematic setup of Gaussian boson sampling. In this example, $K=8$ single-mode squeezed states are injected into a passive linear optical network with $M=10$ optical modes. Then photon number resolving detectors detect the photon number in each output mode. An output pattern $\bar{n}=n_1 n_2 ... n_M$ is generated according to the detected results.}
	\label{fig:1}
\end{figure}

Before detection, the output quantum state of the passive linear optical network is a Gaussian quantum state~\cite{wang_quantum_2007,serafini_quantum_2017,weedbrook_gaussian_2012,simon_quantum-noise_1994}.
A Gaussian state is fully determined by its covariance matrix and displacement vector. Denote operator vector $\hat{\xi} = (\hat{a}_1^\dagger,\dots,\hat{a}_M^\dagger,\hat{a}_1,\dots,\hat{a}_M)^{T}$, where $\hat{a}_i^\dagger$ and $\hat{a}_i$ are the creation and annihilation operators in the $i$th ($i\in\left\{1,2,\dots,M\right\}$) optical mode, respectively.
The covariance matrix $\sigma$ and displacement vector $\bar{r}$ of the Gaussian state $\hat{\rho}$ is defined as
\begin{equation}
	\begin{aligned}
		&\bar{r} = \operatorname{Tr}[\hat{\rho} \hat{\xi}],\\
		&\sigma=\frac{1}{2}\operatorname{Tr}\left[\hat{\rho}\left\{(\hat{\xi}-\bar{r}),(\hat{\xi}-\bar{r})^{\dagger}\right\}\right],\\
	\end{aligned}
\end{equation}
where
\begin{equation}
	\begin{aligned}
		\left\{(\hat{\xi}-\bar{r}),(\hat{\xi}-\bar{r})^{\dagger}\right\} =& (\hat{\xi}-\bar{r}) (\hat{\xi}-\bar{r})^{\dagger} \\
		&+ ((\hat{\xi}-\bar{r}) (\hat{\xi}-\bar{r})^{\dagger})^{\dagger}.
	\end{aligned}
\end{equation}
Notice that $(\hat{\xi}-\bar{r}) (\hat{\xi}-\bar{r})^{\dagger}$ is the outer product of column vector $(\hat{\xi}-\bar{r})$ and row vector $(\hat{\xi}-\bar{r})^{\dagger}$ and $(\hat{\xi}-\bar{r}) (\hat{\xi}-\bar{r})^{\dagger} \ne ((\hat{\xi}-\bar{r}) (\hat{\xi}-\bar{r})^{\dagger})^{\dagger}$ as the operators  in the vector $\hat{\xi}$ might not commute with each other. As an example, assume $M=1$ so that $\hat{\xi}=(\hat{a}_1^\dagger,\hat{a}_1)^T$, we have
\begin{equation}
	(\hat{\xi}-\bar{r})(\hat{\xi}-\bar{r})^{\dagger}=\left(\begin{array}{cc}
		\hat{a}_1^\dagger \hat{a}_1 & \hat{a}_1^\dagger \hat{a}_1^\dagger \\
		\hat{a}_1 \hat{a}_1 & \hat{a}_1 \hat{a}_1^\dagger
	\end{array}\right).
\end{equation}
But, 
\begin{equation}
	((\hat{\xi}-\bar{r})(\hat{\xi}-\bar{r})^{\dagger})^{\dagger}=\left(\begin{array}{cc}
		 \hat{a}_1 \hat{a}_1^\dagger & \hat{a}_1^\dagger \hat{a}_1^\dagger \\
		\hat{a}_1 \hat{a}_1 &  \hat{a}_1^\dagger \hat{a}_1
	\end{array}\right) \ne (\hat{\xi}-\bar{r})(\hat{\xi}-\bar{r})^{\dagger}.
\end{equation}

\hl{Usually, a matrix (say matrix $A$) is used  to calculate the output probability distribution of a GBS process~\cite{kruse_detailed_2019,hamilton_gaussian_2017,quesada_quadratic_2022}.}  Denote that $X_{2M} = \left(\begin{array}{cc}
		0 & I_M \\
		I_M & 0
	\end{array}\right)$, and $I_{2M}$ (or $I_M$) as identity matrix with rank $2M$ (or $M$). The matrix $A$ is fully determined by the output Gaussian sate as follows:
\begin{equation}
		A_{i,j}=\begin{cases}(\tilde{A})_{i, j} & \text { if } i \neq j \\ \tilde{y}_i & \text { if } i=j\end{cases},
\end{equation}
where
\begin{equation}
		\begin{aligned}
		&\sigma_Q=\sigma +I_{2 M} / 2, \\
		&\tilde{A}=X_{2 M}\left(I_{2 M}-\sigma_Q^{-1}\right), \\
		& \tilde{y}= X_{2 M}\sigma_Q^{-1} \bar{r}.
	\end{aligned}
\end{equation}

 The probability of generating an output sample $\bar{n}$ is
\begin{equation}
	\begin{aligned}\label{eq:2_2_1}
		&p(\bar{n})=\frac{\exp \left(-\frac{1}{2} \bar{r}^{\dagger} \sigma_Q^{-1} \bar{r} \right)}{\sqrt{\operatorname{det}(\sigma_Q)}} \frac{\operatorname{lhaf}\left(A_{\bar{n}}\right)}{n_{1} ! \cdots n_{M} !},\\
		\end{aligned}
\end{equation}
where $\operatorname{lhaf}(A)=\sum_{M \in \operatorname{SPM}(n)} \prod_{(i, j) \in M} A_{i, j}$ is the loop Hafnian function of matrix $A$, $\operatorname{SPM}$ is the single-pair matchings~\cite{bjorklund_faster_2019} which is the set of perfect matchings with loops. The matrix $A_{\bar{n}}$ is obtained from $A$ as follows $\bar{n}$~\cite{quesada_quadratic_2022,quesada_exact_2020}: for $\forall i = 1,\dots, M$, if $n_i = 0$, the rows and columns are deleted from matrix A; if $n_i \ne 0$, rows and columns $i$ and $i+M$ are repeated $n_i$ times.
%

\subsection{Continuous variables quantum systems}
If a Gaussian state is input into a linear optical network, the output quantum state is also a Gaussian state.
Denote the unitary operator corresponding to the passive linear optical network as $\hat{U}$. A property of the passive linear optical network is:
\begin{equation}
	\hat{U}^\dagger \hat{\xi} \hat{U} = T \hat{\xi},
\end{equation}
where
\begin{equation}
	T = \left(\begin{array}{cc}
		U & 0 \\
		0 & U^*
	\end{array}\right).
\end{equation} 
The output Gaussian state and the input Gaussian state is related by
\begin{equation}
	\begin{aligned}
		&\sigma =T \sigma_{\text{in}} T^{\dagger},\\
		&\bar{r} = T\bar{r}_{\text{in}},\\
		\label{eq:2_2_2}
	\end{aligned}
\end{equation}
where $\sigma_\text{in}$ and $\bar{r}_{\text{in}}$ are the covariance matrix and displacement vector of the input Gaussian state. 
The SMSS in input mode $i$ has squeezing strength $s_i$ and phase $\phi_i$. For simplicity, assume that $\phi_i = 0$ for $i = 1,\dots,M$. The covariance matrix of the input Gaussian state is
\begin{equation}
	\sigma_{in} = \frac{1}{2}\left(\begin{array}{cc}
		\bigoplus\limits_{i=1}^M \cosh 2s_i & \bigoplus\limits_{i=1}^M \sinh 2s_i\\
		\bigoplus\limits_{i=1}^M \sinh 2s_i & \bigoplus\limits_{i=1}^M \cosh 2s_i 
	\end{array} \right).
\label{eq:2_2_3}
\end{equation}
The matrix $\tilde{A}$ is thus $\tilde{A} = \tilde{B} \bigoplus \tilde{B^*}$, where 
\begin{equation}
	\tilde{B} = U \left(\bigoplus\limits_{i=1}^M \tanh s_i \right) U^T.
	\label{eq:2_2_4}
\end{equation}

If a part of the optical modes in the Gaussian state are measured with photon number resolving detectors, the remaining quantum state will be a non-Gaussian state~\cite{quesada_simulating_2019,su_conversion_2019}.
A Gaussian state can be represented in the following form:
\begin{equation}
	\begin{aligned}
		\hat{\rho} &= \frac{1}{\pi^{2M}} \int \mathrm{d}^2 \vec{\alpha} \int \mathrm{d}^2 \vec{\beta} |\vec{\beta}\rangle \langle\vec{\alpha}| 
		\langle \vec{\beta} \left| \hat{\rho} \right| \vec{\alpha} \rangle       \\
		& = \frac{\mathcal{P}_0}{\pi^{2M}} \int \mathrm{d}^2 \vec{\alpha} \int \mathrm{d}^2 \vec{\beta} |\vec{\beta}\rangle \langle\vec{\alpha}| 
		\mathrm{exp} (-\frac{|\tilde{\lambda}|^2}{2} + \frac{1}{2} \tilde{\lambda}^{\mathrm{T}} \tilde{A} \tilde{\lambda} +\tilde{\lambda}^{\mathrm{T}}\tilde{y}  ),       
	\end{aligned}
\end{equation}
where 
\begin{equation}
	\begin{aligned}
		\tilde{\lambda} &= \left(\beta_{1}^*,\ldots,\beta_{M}^*, \alpha_{1},\dots,\alpha_{M}\right)^{\mathrm{T}},\\
		\mathcal{P}_0  &=\frac{\exp \left(-\frac{1}{2} \bar{r}^{\dagger} \sigma_Q^{-1} \bar{r} \right)}{\sqrt{\operatorname{det}(\sigma_Q)}},\\
		\vec{\alpha} &= \left(\alpha_{1},\dots,\alpha_{M}\right)^{\mathrm{T}},\\
		\vec{\beta} &= \left(\beta_{1},\ldots,\beta_{M}\right)^{\mathrm{T}},\\
	\end{aligned}
\end{equation}
$\beta_{i}$ and $\alpha_i$ for $i=1,\dots,M$ are complex variables.

For our convenience, define the permutation matrix $P$, such that $\tilde{\gamma} = P \tilde{\lambda} = \left(\beta_{1}^*,\alpha_{1},\beta_{2}^*,\alpha_{2},\dots,\beta_{M}^*, \alpha_{M}\right)^{\mathrm{T}}$. Let $\tilde{R} = P^\mathrm{T}\tilde{A}P$, $\tilde{l} = P \tilde{y}$. We then have
\begin{equation}
		\hat{\rho} = \frac{\mathcal{P}_0}{\pi^{2M}} \int \mathrm{d}^2 \vec{\alpha} \int \mathrm{d}^2 \vec{\beta} |\vec{\beta}\rangle \langle\vec{\alpha}| 
		\mathrm{exp} (-\frac{|\tilde{\gamma}|^2}{2} + \frac{1}{2} \tilde{\gamma}^{\mathrm{T}} \tilde{R} \tilde{\gamma} +\tilde{\gamma}^{\mathrm{T}}\tilde{l}  ).  
		\label{eq:2}     
\end{equation}

Suppose the first $N$ modes of an $M$-mode Gaussian state are measured and a sample pattern $\bar{n}=n_1 n_2 \dots n_N$ is observed.
Denote that
\begin{equation}\label{eq:2_1_2}
	\begin{aligned}
		\tilde{\gamma}_h  &= (\tilde{\gamma}_{2N+1}^*,\tilde{\gamma}_{2N},\dots,\tilde{\gamma}_{2M})^{\mathrm{T}},\\
		\tilde{\gamma}_d &= (\tilde{\gamma}_1, \tilde{\gamma}_2, \dots,\tilde{\gamma}_{2N}) ^{\mathrm{T}}, \\
		\tilde{l}_h &= (\tilde{l}_{2N+1},\tilde{l}_{2N}\dots,\tilde{y}_{2M})^T,\\
		\tilde{l}_d &= (\tilde{l}_{1},\tilde{l}_{2}\dots,\tilde{l}_{2N})^T,\\
		\tilde{R}&=\left(\begin{array}{cc}
			\tilde{R}_{dd} & \tilde{R}_{dh}\\
			\tilde{R}_{hd} & \tilde{R}_{hh}
		\end{array}\right),
	\end{aligned}
\end{equation}
where $\tilde{R}_{dd}$ is a $2N\times 2N$ matrix corresponding to modes $1$ to $N$, $\tilde{R}_{hh}$ is a $2(M-N) \times 2(M-N)$ matrix corresponding to modes $N+1$ to $M$, $\tilde{R}_{dh}$ is a $2N \times 2(M-N)$ matrix represents the correlation between modes $1$ to $N$ and modes $N+1$ to $M$.
The remaining non-Gaussian state is~\cite{su_conversion_2019}
\begin{equation}
	\begin{aligned}
		\hat{\rho}^{(\bar{n})}=&\frac{\mathcal{P}_0}{\pi^{2 (M-N)}} \int d^2 \vec{\alpha}_{(N)}  \int d^2 \vec{\beta}_{(N)}  |\vec{\beta}_{(N)} \rangle \langle\vec{\alpha}_{(N)} |  F\left(\tilde{\gamma}_h\right),
	\end{aligned}
	\label{eq:2_1_1}
\end{equation}
where 
\begin{equation}
	\begin{aligned}
		d^2 \vec{\beta}_{(N)} &= d^2 \beta_{N+1} d^2 \beta_{N+2}\dots d^2\beta_{M},\\
		d^2 \vec{\alpha}_{(N)} &= d^2\alpha_{N+1} d^2 \alpha_{N+2}\dots d^2 \alpha_{M},
	\end{aligned}
\end{equation} 
$|\vec{\beta}_{(N)}\rangle = |\beta_{N+1},\beta_{N+2}\dots,\beta_{M}\rangle $ and $|\vec{\alpha}_{(N)}\rangle = |\alpha_{N+1},\alpha_{N+2}\dots,\alpha_{M}\rangle $ are coherent states, and 
\begin{equation}
	\begin{aligned}
		F\left(\tilde{\gamma}_h\right) 
		& =\left.\frac{\mathcal{P}_0}{\bar{n} !} \exp \left(L_2\right) \prod_{k=N+1}^M\left(\frac{\partial^2}{\partial \alpha_k \partial \beta_k^*}\right)^{n_k} \exp \left(L_3\right)\right|_{\tilde{\gamma}_d=0}, \\
		L_2 & =-\frac{1}{2}\left|\tilde{\gamma}_h\right|^2+\frac{1}{2} \tilde{\gamma}_h^{T} \tilde{R}_{h h} \tilde{\gamma}_h+\tilde{\gamma}_h^{T} \tilde{l}_h, \\
		L_3 &= -\left|\tilde{\gamma}_d\right|^2+\frac{1}{2} \tilde{\gamma}_d^{T} \tilde{R}_{d d} \tilde{\gamma}_d+\tilde{\gamma}_d^{T} \tilde{l}_d+\tilde{\gamma}_d^{T} \tilde{R}_{d h} \tilde{\gamma}_h .
	\end{aligned}
\end{equation}

If all optical modes of the Gaussian state are measured (N=M), then Eq.~(\ref{eq:2_1_1}) gives the probability of obtaining the output sample $\bar{n}$, which is equivalent to Eq.~(\ref{eq:2_2_1}).

\subsection{Limited connectivity} \label{sec:2_3}
In the Gaussian Boson Sampling protocol~\cite{hamilton_gaussian_2017,kruse_detailed_2019}, the transform matrix $U$ of the passive linear optical network should be randomly chosen from the Haar measure.
The circuit depth needed to realize an arbitrary unitary transform in a passive linear optical network is $O(M)$, where $M$ is the number of optical modes~\cite{russell_direct_2017,reck_experimental_1994,clements_optimal_2016,go_exploring_2023}. 
However, due to the photon loss, the circuit depth of the optical network might not be deep enough to meet the requirements of the full connectivity and the global Haar-random unitary~\cite{russell_direct_2017}. This is because photon loss rate $\epsilon$ will increase exponentially with the depth of the optical network, i.e., $\epsilon = \epsilon_0^D$, where $\epsilon_0$ is the photon loss of each layer of the optical network. If the photon loss rate is too high, the quantum advantage result of GBS experiments will be destroyed~\cite{qi_regimes_2020,martinez-cifuentes_classical_2023,oh_tensor_2023}.

\begin{figure}
	\centering
	\includegraphics[width=0.6\linewidth]{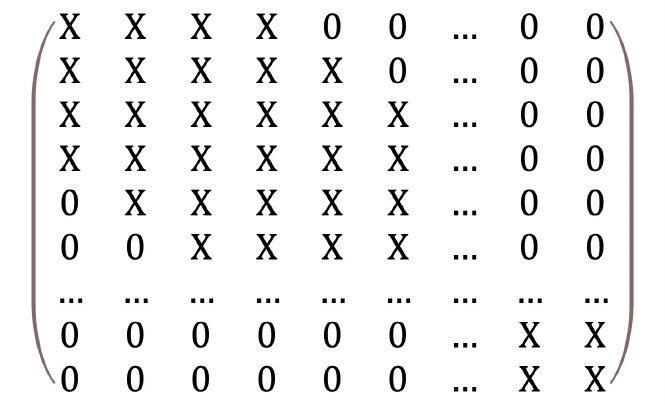}
	\caption{A symmetric matrix with a bandwidth structure. X represents an arbitrary non-zero entry in the matrix. The bandwidth in this example is $w = 3$.}
	\label{fig:2}
\end{figure}

The shallow circuit depth lead to a limited connected interferometer~\cite{russell_direct_2017,qi_efficient_2022,oh_classical_2022}. Assume that the beam splitters in the optical network is local, which means they can only connect \hl{the adjacent optical modes} as shown in Fig.~\ref{fig:1}.
If $D<M$, the transform matrix $U$ of the passive linear optical network will have a bandwidth structure~\cite{qi_efficient_2022}, i.e., 
\begin{equation}
	\left|U_{i,j}\right|=0 \quad \text{for} \ \left|i-j \right|>w_U,
\end{equation}
and
\begin{equation}
	\left|U_{i,j}\right|\ne 0 \quad \text{for} \ \left|i-j \right| \le w_U,
\end{equation} 
where $w_U= D$ is the bandwidth of the transform matrix $U$. An example of a matrix with bandwidth structure is shown in Fig.~\ref{fig:2}.
According to Eq.~(\ref{eq:2_2_2}) and~(\ref{eq:2_2_3}), the covariance matrix $\sigma$ of the output Gaussian state will have a bandwidth structure.
Thus, if beam splitters are local, the circuit depth $D$ must be no less than $M/2$ to reach the full connectivity.
Note that, according to Eq.~(\ref{eq:2_2_4}), matrix $\tilde{B}$ (and $A$ in Eq.~(\ref{eq:2_2_1})) has bandwidth
\begin{equation}
		w = \left\{ \begin{array}{c}
			2 w_U -1 \quad \text{for} \ w_U < M/2 \\
			M - 1 \quad \text{for} \ w_U \ge M/2
			\end{array}\right. .
\end{equation}

Recently, a scheme known as ``high-dimensional GBS" has been proposed~\cite{deshpande_quantum_2022}. This scheme suggests that by interfering non-adjacent optical modes, the connectivity can be improved while maintaining a relatively shallow circuit depth. 
However, due to the limited circuit depth, the transformation matrix in this scheme cannot represent an arbitrary unitary matrix. 
Consequently, the transformation matrix deviates from the global Haar-random unitary. 
To address this, the scheme introduces a local Haar-random unitary assumption, which says that the transformation matrix corresponding to each individual beam splitter is randomly selected from the Haar measure. 
Under this local Haar-random unitary assumption, Ref.~\cite{oh_classical_2022} demonstrates that when the circuit depth is too shallow, the high dimensional GBS process can be approximate by a limited connected GBS process with a small error.



As a result of the shallow circuit depth and the deviation from the Haar measure,  a speed-up can be realized in simulating the corresponding GBS process~\cite{qi_efficient_2022,oh_classical_2022}. 
The speed-up is attributed to the faster computation of the loop Hafnian for matrices with bandwidth as compared to the computation of general matrices.

\subsection{Classical simulation of GBS} \label{sec:2_4}
To date, the most efficient classical simulation method for simulating a general GBS process has been presented in Ref.~\cite{quesada_quadratic_2022}.
This classical algorithm, which samples from an $M$-modes Gaussian state $\hat{\rho}$, operates as follows.

\begin{enumerate}
	\item If $\hat{\rho}$ is a mixed state, it can be decomposed as a classical mixture of pure displaced Gaussian states. We randomly select a pure displaced Gaussian state based on this classical mixture. This can be done in polynomial time~\cite{serafini_quantum_2017,quesada_quadratic_2022}.  Its covariance matrix and displacement vector are denoted as $\sigma$ and $\bar{r}$, respectively. 
	\item If $\hat{\rho}$ is a pure state, denote its covariance matrix and displacement vector as $\sigma$ and $\bar{r}$, respectively.
\end{enumerate}
For $k=1,\dots,M$:
\begin{enumerate}
	\setcounter{enumi}{1}
	\item Drawn a sample $\bar{\alpha}^k = \left(a_{k+1},\dots,\alpha_{M}\right)$ from the probability distribution:
	\begin{equation}
		p(\bar{\alpha}) = \left\langle \bar{\alpha}^k \left|\hat{\rho}_{d}\right|  \bar{\alpha}^k  \right\rangle,
	\end{equation}
where $\hat{\rho}_{d} = \operatorname{Tr}_{1\mbox{-}k} \left[\hat{\rho}\right]$, and $\operatorname{Tr}_{1\mbox{-}k} \left[\cdot\right]$ is the partial trace from mode $1$ to mode $k$. This is equivalent to measure the modes $k+1$ to $M$ by heterodyne measurements.  
	\item Compute the conditional covariance matrix $\sigma_{1\mbox{-}k}^{(\bar{\alpha}^k)}$ and displacement $\bar{r}_{1\mbox{-}k}^{(\bar{\alpha}^k)}$ of the remaining Gaussian state. 
	If $k=M$, let $\sigma_{1\mbox{-}k}^{(\bar{\alpha}^k)} = \sigma$ and $\bar{r}_{1\mbox{-}k}^{(\bar{\alpha}^k)}=\bar{r}$.
	Notice that the conditional quantum state in modes $1,\dots,k$ is still Gaussian if modes $k+1,\dots,M$ is measured by heterodyne measurements~\cite{serafini_quantum_2017}.
	\item Given a cutoff $N_{\text{max}}$, use Eq.~(\ref{eq:2_2_1}) with $\sigma_{1\mbox{-}k}^{(\bar{\alpha}^k)}$ and  $\bar{r}_{1\mbox{-}k}^{(\bar{\alpha}^k)}$ to calculate $p\left(n_1,\dots,n_k\right)$ for $n_k = 0,1,\dots,N_{\text{max}}$. 
	\item Sample $n_k$ from:
	\begin{equation}
		p(n_k) = \frac{p\left(n_1,\dots,n_k\right)}{p\left(n_1,\dots,n_{k-1}\right)}.
		\label{eq:3}
	\end{equation}
\end{enumerate}

The most time-consuming part for the above classical simulation method is to calculate the probability $p\left(n_1,\dots,n_k\right)$ of an output sample pattern $n_1,\dots,n_M$. According to Eq.~(\ref{eq:2_2_1}), we have
\begin{equation}
	\begin{aligned}
		p\left(n_1,\dots,n_k\right) =& \frac{\exp \left(-\frac{1}{2} \bar{r}_{\mathcal{A}}^{(\mathcal{B})\dagger} (\sigma_Q)_{\mathcal{A}}^{(\mathcal{B})-1} \bar{r}_{\mathcal{A}}^{(\mathcal{B})} \right)}{\sqrt{\operatorname{det}\left((\sigma_Q)_{\mathcal{A}}^{(\mathcal{B})}\right) }} \\
		&\times \frac{\operatorname{lHaf}\left(\left(A_{\mathcal{A}}^{(\mathcal{B})}\right)_{\bar{n}}\right)}{n_{1} ! \cdots n_{M} !},\\
	\end{aligned}
\end{equation}
where the $\bar{r}_{\mathcal{A}}^{(\mathcal{B})}$, $(\sigma_Q)_{\mathcal{A}}^{(\mathcal{B})}$ and $A_{\mathcal{A}}^{(\mathcal{B})}$ are the corresponding conditional matrices for subsystem contains modes $1$ to $k$ (denoted as $\mathcal{A}$) when modes $k+1$ to $M$ are measured by heterodyne measurements with outcome denoted as $\mathcal{B}$. $A_{\mathcal{A}}^{(\mathcal{B})}$ can be computed by the following equation:
\begin{equation}
	A_{\mathcal{A}}^{(\mathcal{B})} = \begin{cases}(\tilde{A}_{\mathcal{A}}^{(\mathcal{B})})_{i, j} & \text { if } i \neq j \\ (\tilde{y}_{\mathcal{A}}^{(\mathcal{B})})_i & \text { if } i=j\end{cases}.
\end{equation}
Note that if the Gaussian state $\hat{\rho}$ is a pure state, the conditional Gaussian state is still a pure state. So, we have $\tilde{A}_{\mathcal{A}}^{(\mathcal{B})} = \tilde{B}_{\mathcal{A}} \bigoplus \tilde{B}_{\mathcal{A}}^*$. As pointed in Ref.~\cite{oh_classical_2022}, we have $\tilde{B}_{\mathcal{A}} = [U (\bigoplus\limits_{i} \tanh s_i ) U^T ]_{\mathcal{A}} $. This shows that $\tilde{B}_{\mathcal{A}}$ has the same bandwidth structure as $\tilde{B}$. A classical algorithm to calculate the loop Hafnian of a banded matrix is thus needed for the classical simulation of GBS with limited connectivity.

\section{simulate the sampling process with limited connectivity}\label{sec:3}

\subsection{Loop Hafnian algorithm for banded matrices}\label{sec:3_1}
An algorithm to calculate the loop Hafnian of an $n \times n$ symmetric matrix with bandwidth $w$ in time $O\left(nw 4^w\right)$ is given in Ref.~\cite{qi_efficient_2022}.
Then, a faster algorithm with time complexity $O\left(n w_t^2 2^{w_t}\right)$ to calculate the loop Hafnian of an $n\times n$ symmetric matrix is proposed in Ref.~\cite{oh_classical_2022}, where $w_t$ is the treewidth of the graph corresponding to the matrix~\cite{cifuentes_efficient_2016}. For a banded matrix, the smallest treewidth $w_t$ is equal to the bandwidth $w$, i.e., $w = w_t$. So the time complexity for this algorithm to calculate the loop Hafnian of an $n\times n $ symmetric matrix with bandwidth $w$ is $O\left(n w^2 2^w\right)$. 
Here we show that the loop Hafnian of an $n \times n$ symmetric matrix with bandwidth $w$ can be calculated in $O\left( n w 2^w \right)$. 

Our algorithm to calculate the loop Hafnian for banded matrices is outlined as follows.

\textbf{Algorithm.}
To calculate the loop Hafnian of an $n\times n$ symmetric matrix $B$ with bandwidth $w$:\\
\begin{enumerate}
	\item Let $C_{\emptyset}^0 = 1$.
\end{enumerate}
For $t=1,\dots,n$:
\begin{enumerate}
\setcounter{enumi}{1}
	\item Let $t_w = \min\left(t+w,n\right)$, and $P(\{t+1,\dots,t_w\})$ be the set of all subsets of $\left\{t+1,\dots,t_w\right\}$.
	\item For every $Z^t \in P(\{t+1,\dots,t_w\})$ satisfying $Z^t \ne \emptyset$ and $\left|Z^t\right| \le \min \left(t,w\right)$, let
	\begin{equation}
		\quad \quad C^t_{Z^t} = \sum_{x \in Z^t }  B_{t,x} C^{t-1}_{Z^t \backslash \{x\}} + C^{t-1}_{Z^t \cup \left\{t\right\}} + B_{t,t} C^{t-1}_{Z^t},\label{eq:3_1_3}
	\end{equation}
	and if $t_w \in Z^t$, then
	\begin{equation}
		C^t_{Z^t} =  B_{t,t_w} C^{t-1}_{Z^t \backslash \{t_w\}}.\label{eq:3_1_4}
	\end{equation}
	During the above iterations, if $C^{t-1}_{\{\dots\}}$  is not given in the previous steps, it is treated as 0. 
	\item Let 
	\begin{equation}
		C^t_\emptyset = B_{t,t}C^{t-1}_\emptyset + C^{t-1}_{\{t\}}. \label{eq:3_1_5}
	\end{equation}
\end{enumerate}
The loop Hafnian of matrix $B$ is obtained in the final step $t = N$ by
\begin{equation}
	\operatorname{lhaf}\left( B \right) = C^n_{\emptyset}.
\end{equation}
 The time complexity of this algorithm is $O(nw 2^{w})$ as shown in theorem~\ref{thm:1}.

\begin{theorem}\label{thm:1}
	Let $B$ be an $n\times n$ symmetric matrix with bandwidth $w$. Then its loop Hafnian can be calculated in $O(nw 2^{w})$.
\end{theorem}

\begin{proof}
	As shown in our algorithm, the number of coefficients ($C^t_{Z_t}$, $C^t_\emptyset$ and $C^t_{\{t\}}$)  needed to be calculated for each $t \in \{1,\dots,n\}$  is at most $2^{w}$. 
	As shown in Eq.~(\ref{eq:3_1_3})~(\ref{eq:3_1_4})~(\ref{eq:3_1_5}), in each iteration, we need $O(w)$ steps to calculate each coefficient $C^{t}_{Z^t}$. So, for each $t$, the algorithm takes $O(w 2^{w})$ steps. The overall cost is thus $O\left(n w 2^{w} \right)$.

\end{proof}

An example of calculating a $4 \times 4$ matrix with bandwidth $w = 1$ using our algorithm can be found in Appendix~\ref{ap:1}.

Note that, our algorithm for calculating loop Hafnian function of matrices with bandwidth structure can be easily \hl{extended} to cases where the matrices is sparse (but not banded). A description of this can be found in Appendix~\ref{ap:2}

By combining our algorithm with the classical sampling techniques described in Ref.~\cite{quesada_quadratic_2022}, as summarized in Sec.~\ref{sec:2_4}, a limited connected GBS process with a bandwidth $w$ can be simulated in $O(M nw2^w)$ time, where $M$ represents the number of optical modes and $n$ denotes the maximum total photon number of the output samples. A proof of this statement is presented in Appendix~\ref{ap:3}.

\subsection{Validity of the algorithm }\label{sec:3_2}
\hl{By sequentially computing the states (usually non-Gaussian) of the remaining $M-k$ modes with the measurement outcome of modes $1$ to $k$ ($k=1,\dots,M$) in photon number basis, we can demonstrate the validity of the algorithm introduced in Sec.~\ref{sec:3_1}.
Recalling matrix $\tilde{R}$ defined in Eq.~(\ref{eq:2}), we denote the bandwidth of matrix $\tilde{R}$ by $w$.} According to the definition of $\tilde{R}$ and $\tilde{l}$, we know that
\begin{equation}
	\mathrm{lhaf}(R) = \mathrm{lhaf}(A),
\end{equation}
where
\begin{equation}
	R_{i,j} = \begin{cases}
		(\tilde{R})_{i, j} & \text{if } i \neq j \\
		\tilde{l}_i & \text{if } i = j
	\end{cases}.
\end{equation}
Assuming that the measurement results are $n_i = 1$ for every $i = 1, \dots, M$, we obtain
\begin{equation}
	p(11 \dots 1) = \mathcal{P}_0 \mathrm{lhaf}(R).
\end{equation}
\hl{Although we make the assumption that $R$ corresponds to a Gaussian state, the subsequent proof remains valid in the more general case as discussed in Appendix~\ref{ap:4}.}


 For convenience, define $\tilde{\gamma}_{d_i} = (\beta_i^*,\alpha_i)^\mathrm{T}$, $\tilde{\gamma}_{h_i} = (\beta_{i+1}^*,\alpha_{i+1},\dots,\beta_{M}^*,\alpha_{M})^\mathrm{T}$, $\tilde{R}_{dd}^i$ as a $2\times 2$ matrix corresponding to mode $i$, $\tilde{R}_{hh}^i$ as a $2(M-i) \times 2(M-i)$ matrix corresponding to modes $i+1$ to $M$, $\tilde{R}_{dh}^i$ as a $2 \times 2(M-i)$ matrix represents the correlation between mode $i$ and modes $i+1$ to $M$.

According to Eq.~(\ref{eq:2_1_1}), after measuring the mode $1$ in photon number basis, the remaining non-Gaussian quantum state is:
\begin{small}
\begin{equation}
	\begin{aligned}
			\hat{\rho}^{(n_1)} = & \frac{\mathcal{P}_0}{\pi^{2(M-1)}} \int \mathrm{d}^2 \vec{\alpha}_{(1)} \int \mathrm{d}^2 \vec{\beta}_{(1)}
			|\vec{\beta}_{(1)}\rangle \langle \vec{\alpha}_{(1)} |  \\
			&\times \left. \frac{\partial^2}{\partial \alpha_1 \partial \beta_1^*} \mathrm{exp} (\frac{1}{2} \tilde{\gamma}_{d_1}^\mathrm{T} \tilde{R}_{dd}^1 \tilde{\gamma}_{d_1} +\tilde{\gamma}_{d_1}^\mathrm{T} \tilde{R}_{dh}^1 \tilde{\gamma}_{h_1} + \tilde{\gamma}_{d_1}^\mathrm{T}\tilde{l}_{d_1})\right|_{\gamma_{d_1}=0}\\
			& \times \mathrm{exp} (-\frac{1}{2}|\tilde{\gamma}_{h_1}|^2 + \frac{1}{2} \tilde{\gamma}_{h_1}^\mathrm{T} \tilde{R}_{hh}^1 \tilde{\gamma}_{h_1} + \tilde{\gamma}_{h_1}^\mathrm{T}\tilde{l}_{h_1}).\\
	\end{aligned}
\end{equation}
\end{small}
We then have
\begin{small}
\begin{equation}
	\begin{aligned}
		&\left. \frac{\partial^2}{\partial \alpha_1 \partial \beta_1^*} \mathrm{exp} (\frac{1}{2} \tilde{\gamma}_{d_1}^\mathrm{T} \tilde{R}_{dd}^1 \tilde{\gamma}_{d_1} + \tilde{\gamma}_{d_1}^\mathrm{T} \tilde{R}_{dh}^1 \tilde{\gamma}_{h_1} + \tilde{\gamma}_{d_1}^\mathrm{T}\tilde{l}_{d_1})\right|_{\gamma_{d_1}=0}  \\
		=& (\tilde{R}_{dd}^1 )_{1,2} + (\tilde{l}_{d_1})_1(\tilde{l}_{d_1})_2 + \sum_{j,k=1}^{M} (\tilde{R}_{dh}^1)_{1,j} (\tilde{R}_{dh}^1)_{2,k} (\tilde{\gamma}_{h_1})_j
(\tilde{\gamma}_{h_1})_l \\
=& \tilde{R}_{12} + \tilde{l}_1 \tilde{l}_2 + + \sum_{j,k=1}^{M} \tilde{R}_{1,j+2} \tilde{R}_{2,k+2} \tilde{\gamma}_{j+2}  \tilde{\gamma}_{k+2}\\
		=& C^2_{\emptyset} + \sum_{Z^2} C^2_{Z^2}  \prod_{x \in Z^2} \tilde{\gamma}_x + \sum_{x \in \{3,\dots,M\} } D^2_x \tilde{\gamma}_x^2,
	\end{aligned}
\end{equation}
\end{small}
where $D^2_x$ is the coefficient for $\tilde{\gamma}_x^2$. 

Next, measuring mode $2$. The remaining non-Gaussian quantum state is:
\begin{small}
\begin{equation}
	\begin{aligned}
		\hat{\rho}^{(n_1n_2)} = & \frac{\mathcal{P}_0}{\pi^{2(M-2)}} \int \mathrm{d}^2 \vec{\alpha}_{(2)} \int \mathrm{d}^2 \vec{\beta}_{(2)}
		|\vec{\beta}_{(2)}\rangle \langle \vec{\alpha}_{(2)} |  \\
		& \times \mathrm{exp} (-\frac{1}{2}|\tilde{\gamma}_{h_2}|^2 + \frac{1}{2} \tilde{\gamma}_{h_2}^\mathrm{T} \tilde{R}_{hh}^2 \tilde{\gamma}_{h_2} + \tilde{\gamma}_{h_2}^\mathrm{T}\tilde{l}_{h_2})\\
		&\times  \frac{\partial^2}{\partial \alpha_2 \partial \beta_2^*} [\mathrm{exp} (\frac{1}{2} \tilde{\gamma}_{d_2}^\mathrm{T} \tilde{R}_{dd}^2 \tilde{\gamma}_{d_2} +\tilde{\gamma}_{d_2}^\mathrm{T} \tilde{R}_{dh}^2 \tilde{\gamma}_{h_2} + \tilde{\gamma}_{d_2}^\mathrm{T}\tilde{l}_{d_2})\\
		&\left. \times (C^2_{\emptyset} + \sum_{Z^2} C^2_{Z^2}  \prod_{x \in Z^2} \tilde{\gamma}_x + \sum_{x \in \{3,\dots,M\} } D^2_x \tilde{\gamma}_x^2) \left]\right.\right|_{\gamma_{d_2}=0}.\\
	\end{aligned}
\end{equation}
\end{small}
We have 
\begin{small}
\begin{equation}
	\begin{aligned}
		& \frac{\partial^2}{\partial \alpha_2 \partial \beta_2^*} [\mathrm{exp} (\frac{1}{2} \tilde{\gamma}_{d_2}^\mathrm{T} \tilde{R}_{hh}^2 \tilde{\gamma}_{d_2} +\tilde{\gamma}_{d_2}^\mathrm{T} \tilde{R}_{dh}^2 \tilde{\gamma}_{h_2} + \tilde{\gamma}_{d_2}^\mathrm{T}\tilde{l}_{d_2})\\
		&\left. \times (C^2_{\emptyset} + \sum_{Z^2} C^2_{Z^2}  \prod_{x \in Z^2} \tilde{\gamma}_x + \sum_{x \in \{3,\dots,M\} } D^2_x \tilde{\gamma}_x^2) \left]\right.\right|_{\gamma_{d_2}=0}\\
		=& C^4_{\emptyset}  +\sum_{Z^4} C^4_{Z^4} \prod_{x \in Z^4} \tilde{\gamma}_x +\sum_{x_1,x_2,x_3 \in \{5,\dots,M\}} D^4_{x_1,x_2,x_3}  \tilde{\gamma}_{x_1}^2 \tilde{\gamma}_{x_2} \tilde{\gamma}_{x_3}	,\\
	\end{aligned}
\end{equation}
\end{small}
where $D^4_{x_1,x_2,x_3} $ is the coefficient for $\tilde{\gamma}_{x_1}^2 \tilde{\gamma}_{x_2} \tilde{\gamma}_{x_3}$.

Repeating this procedure to measure mode $1$ to mode $M$, we eventually find that
\begin{equation}
	\hat{\rho}^{(\bar{n} = 11\dots1)} = \mathcal{P}_0 C^{2M}_\emptyset = p(\bar{n} = 11\dots1) = \mathcal{P}_0 \mathrm{lhaf}(R) .
\end{equation}
Thus we prove
\begin{equation}
	C^{2M}_\emptyset = \mathrm{lhaf}(R).
\end{equation}
This demonstrates the validity of the algorithm given in Sec.~\ref{sec:3_1}.

\hl{Intuitively, as shown in the previous derivation, after measuring the first $t$ modes of the output Gaussian state, the state of the remaining modes is:}
\begin{equation}
	\begin{aligned}
		\hat{\rho}^{(n_1\dots n_t)} = & \frac{\mathcal{P}_0}{\pi^{2(M-t)}} \int \mathrm{d}^2 \vec{\alpha}_{(t)} \int \mathrm{d}^2 \vec{\beta}_{(t)}
		|\vec{\beta}_{(t)}\rangle \langle \vec{\alpha}_{(t)} |  \\
		& \times \mathrm{exp} \left(-\frac{1}{2}|\tilde{\gamma}_{h_t}|^2 + \frac{1}{2} \tilde{\gamma}_{h_t}^\mathrm{T} \tilde{R}_{hh}^t \tilde{\gamma}_{h_t} + \tilde{\gamma}_{h_t}^\mathrm{T}\tilde{l}_{h_t}\right)  \\
		& \times \mathrm{Poly}^t \left(R,\tilde{\gamma}_{h_t} \right)
		,\\
	\end{aligned}
\end{equation}
\hl{where $\mathrm{Poly}^t \left(R,\tilde{\gamma}_{h_t} \right)$ represents the sum of polynomial terms formed by complex variables in $\tilde{\gamma}_{h_t}$ and}
	\begin{equation}
		\label{eq:41}
		\begin{aligned}
			\mathrm{Poly}^t (R,\tilde{\gamma}_{h_t}) =& C^{2t}_\emptyset + \sum_{Z^{2t}} C^{2t}_{Z^{2t}} \prod_{x \in Z^{2t}} \tilde{\gamma}_x \\
			&+\sum_{x_1,\dots,x_{2t-1} \in \{2t+1,\dots,M\}} D^{2t}_{x_1,\dots,x_{2t-1}} \\
			& \times  \tilde{\gamma}_{x_1}^2 \tilde{\gamma}_{x_2} \dots \tilde{\gamma}_{x_{2t-1}}.
		\end{aligned}
	\end{equation}
 \hl{When all modes of the output state are measured, we have $t=M$ for Eq.~(\ref{eq:41}), and hence $\mathrm{Poly}^{M} (R,\tilde{\gamma}_{h_M})$ will not contain any complex variables, ensuring that the constant term $C^{2M}_\emptyset$ equals to $\mathrm{lhaf}(R)$. Terms $\{\prod_{x \in Z^{2t}} \tilde{\gamma}_x\}$ of different $Z^{2t}$ influence the constant term $C^{2M}_\emptyset$ in the subsequent computing.  High-order terms containing $\gamma_x^k$ with $k\ge 2$ do not affect the final result. Consequently, the iteration of $C^t_{Z^t}$ appearing in step 3 of our algorithm is required from $t=1$ to $t=2M$ to eventually determine the value of $\mathrm{lhaf}(R) = C^{2M}_\emptyset$.}

\section{Summary}\label{sec:4}
We present an algorithm to calculate the loop Hafnian of banded matrices. This algorithm runs in $O\left(nw2^w\right)$ for $n\times n$ symmetric matrices with bandwidth $w$. Our result is better than the prior art result of  $O\left(nw^2 2^w\right)$~\cite{oh_classical_2022}. Our classical algorithm is helpful on clarifying how limited connectivity reduces the computational resources required for classically simulating GBS processes and tightening the boundary of quantum advantage in GBS problem.



\appendix

\section{Example}\label{ap:1}
Here we give an example of computing the loop Hafnian of a $4 \times 4$ symmetric matrix with bandwidth $w = 1$ by the algorithm given in Sec.~\ref{sec:3_1}.
We consider the following case:
\begin{equation}
	\begin{aligned}
	\operatorname{lhaf}\left(A\right) =& \operatorname{lhaf}\left(\left[\begin{array}{cccc}
		A_{1,1} & A_{1,2} & 0 & 0\\
		A_{2,1} & A_{2,2} & A_{2,3} & 0\\
		0 & A_{3,2} & A_{3,3} & A_{3,4}\\
		0 & 0 & A_{4,3} & A_{4,4}\\
	\end{array}\right]\right) \\
	=& A_{1,1} A_{2,2} A_{3,3} A_{4,4} + A_{1,1} A_{2,2} A_{3,4} \\
	&+ A_{1,1} A_{2,3} A_{3,4} + A_{1,2} A_{3,3} A_{4,4} \\
	&+ A_{1,2}A_{3,4}.
	\end{aligned}
\end{equation}

The algorithm given in Sec.~\ref{sec:3_1} works as follows: 

Step 1.

We have $C^{0}_{\emptyset}=1$ .

Step 2.

According to Eq.~(\ref{eq:3_1_3})~(\ref{eq:3_1_4})~(\ref{eq:3_1_5}), we have $C^{1}_{\left\{2\right\}}=A_{1,2}$ and $C^{1}_{\emptyset}=A_{1,1}$.

Step 3.

According to Eq.~(\ref{eq:3_1_3})~(\ref{eq:3_1_4})~(\ref{eq:3_1_5}), we have $C^{2}_{\left\{3\right\}}=A_{2,3} C^{1}_{\emptyset} =A_{1,1} A_{2,3}  $ and $C^{2}_{\emptyset}=A_{2,2} C^{1}_{\emptyset} + C^{1}_{\left\{2\right\}} =A_{1,1} A_{2,2} + A_{1,2} $.

Step 3.

According to Eq.~(\ref{eq:3_1_3})~(\ref{eq:3_1_4})~(\ref{eq:3_1_5}), we have $C^{3}_{\left\{4\right\}}=A_{3,4} C^{2}_{\emptyset} =A_{1,1} A_{2,2} A_{3,4} + A_{1,2}A_{3,4}  $ and $C^{3}_{\emptyset}=A_{3,3} C^{2}_{\emptyset} + C^{2}_{\left\{3\right\}} =A_{1,1} A_{2,2} A_{3,3} + A_{1,2} A_{3,3} + A_{1,1} A_{2,3}$.

Step 4.

According to Eq.~(\ref{eq:3_1_5}), we have $C^{4}_{\emptyset}=A_{4,4} C^{3}_{\emptyset} + C^{3}_{\left\{4\right\}}=A_{1,1} A_{2,2} A_{3,3} A_{4,4} + A_{1,2} A_{3,3} A_{4,4} + A_{1,1} A_{2,3} A_{4,4} + A_{1,1} A_{2,2} A_{3,4} + A_{1,2}A_{3,4}$.

We can see that $C^{4}_{\emptyset} = \operatorname{lhaf} \left(A\right)$.

\section{Loop Hafnian of sparse matrices}\label{ap:2}
It is easy to find that, with slight modifications, the algorithm given in Sec.~\ref{sec:3_1} can be used to compute the loop Hafnian of symmetric sparse matrices.
Denote the largest number of zero-valued entries in each rows of a sparse matrix $B$ as $w$. 
The modified algorithm is outlined as follows.\\
\textbf{Algorithm.}
To calculate the loop Hafnian of an $n\times n$  symmetric sparse matrix $B$ with at most $w$ non-zero entries in each rows:\\
\begin{enumerate}
	\item Let $C_{\emptyset}^0 = 1$.
\end{enumerate}
For $t=1,\dots,n$:
\begin{enumerate}
	\setcounter{enumi}{1}
	\item Let $w_t$ be the number of non-zero entries in row $t$. Let $\{t_1,\dots,t_{w_t}\}$ be the columns that $B_{t,t_i} \ne 0$ for $i = 1,\dots,w_t$, and $P(\{t_1,\dots,t_{w_t}\})$ be the set of all subsets of $\left\{t_1,\dots,t_{w_t}\right\}$.
	\item For every $Z^t \in P(\{t_1,\dots,t_{w_t}\})$ satisfying $Z^t \ne \emptyset$ and $\left|Z^t\right| \le \min \left(t,w_t\right)$, if $t_{w_t} \in Z^t$, then
	\begin{equation}
		\quad \quad C^t_{Z^t} = \sum_{x \in Z^t }  B_{t,x} C^{t-1}_{Z^t \backslash \{x\}} + C^{t-1}_{Z^t \cup \left\{t\right\}} + B_{t,t} C^{t-1}_{Z^t},
	\end{equation}
	and if $t_{w_t} \in Z^t$, then
	\begin{equation}
		C^t_{Z^t} =  B_{t,t_{w_t}} C^{t-1}_{Z^t \backslash \{t_{w_t}\}}.
	\end{equation}
	During the above iterations, if $C^{t-1}_{\{\dots\}}$  is not given in the previous steps, it is treated as 0.  
	\item Let 
	\begin{equation}
		C^t_\emptyset = B_{t,t}C^{t-1}_\emptyset + C^{t-1}_{\{t\}}.
	\end{equation}
\end{enumerate}
The loop Hafnian of matrix $B$ is obtained in the final step $t = N$ by
\begin{equation}
	\operatorname{lhaf}\left( B \right) = C^n_{\emptyset}.
\end{equation}
The time cost for this algorithm is still $O\left( n w 2^w \right)$.

\section{The time complexity for simulation a GBS process with limited connectivity}\label{ap:3}
The classical simulation process is similar to that in Ref.~\cite{quesada_quadratic_2022}, as summarized in Sec.~\ref{sec:2_4}.
In this process, we sequentially sample $n_k$ for $k=1,\dots,M$ according to Eq.~(\ref{eq:3}). 
As we shown in Sec.~\ref{sec:3_2}, the computation of Eq.~(\ref{eq:3}) takes at most $O\left( n w 2^w \right)$ steps, assuming that $B$ has a bandwidth $w$, where $n = \sum_{i=1}^{M} n_i$. 
This is scaled up by at most the total number of modes $M$. Thus the time complexity for simulating such a GBS process is $O\left( M n w 2^w \right)$.

\section{Validity of the loop Hafnian algorithm for arbitrary matrix $R$}\label{ap:4}
For an arbitrary even matrix $R$, as shown in Ref.~\cite{hamilton_gaussian_2017,kruse_detailed_2019}, we have
\begin{equation}
	\left.\prod_{i=1}^M\left(\frac{\partial^2}{\partial \alpha_i \partial \beta_i^*}\right) \exp \left[\frac{1}{2} \tilde{\gamma}^\mathrm{T} \tilde{R} \tilde{\gamma} + \tilde{\gamma}^\mathrm{T} \tilde{l} \right]\right|_{\tilde{\gamma}=0} = \mathrm{lhaf}(R).
\end{equation}
If we calculate the partial derivative for $i=1,\dots,t$, we have
\begin{equation}
	\begin{aligned}
			\mathrm{lhaf}(R) = &\prod_{i=t+1}^M\left(\frac{\partial^2}{\partial \alpha_i \partial \beta_i^*}\right)  \mathrm{Poly}^t (R,\tilde{\gamma}_{h_t}) \\
			&\times  \left.\mathrm{\exp} \left[\frac{1}{2} \tilde{\gamma}_{h_t}^\mathrm{T} \tilde{R}^t_{hh} \tilde{\gamma}_{h_t} 
			+\tilde{\gamma}_{h_t}^\mathrm{T} \tilde{l}_{h_t} \right]\right|_{\tilde{\gamma}_{h_t}=0}.\\
				\label{eq:ap_4_2}
	\end{aligned}
\end{equation}

For an arbitrary odd matrix $R$ with rank $2M+1$, we have
\begin{equation}
	\left.\prod_{i=1}^M\left(\frac{\partial^2}{\partial \alpha_i \partial \beta_i^*}\right) \frac{\partial}{\partial \beta_{M+1}^*} \exp \left[\frac{1}{2} \tilde{x}^\mathrm{T} \tilde{R} \tilde{x} + \tilde{x}^\mathrm{T} \tilde{z} \right]\right|_{\tilde{\gamma}=0} = \mathrm{lhaf}(R),
\end{equation}
where $\tilde{x} = (\tilde{\gamma}^{\mathrm{T}},\beta_{M+1}^*)^{\mathrm{T}}$, $\tilde{z} = (\tilde{l}^{\mathrm{T}},\tilde{l}_{2M+1})^{\mathrm{T}}$.
If we calculate the partial derivative for $i=1,\dots,t$, we have
\begin{equation}
	\begin{aligned}
		\mathrm{lhaf}(R) = &\prod_{i=t+1}^M\left(\frac{\partial^2}{\partial \alpha_i \partial \beta_i^*}\right) \frac{\partial}{\partial \beta_{M+1}^*}  \mathrm{Poly}^t (R,\tilde{\gamma}_{h_t}) \\
		&\times  \left.\mathrm{\exp} \left[\frac{1}{2} \tilde{x}_{h_t}^\mathrm{T} \tilde{R}^t_{hh} \tilde{x}_{h_t} 
		+\tilde{x}_{h_t}^\mathrm{T} \tilde{z}_{h_t} \right]\right|_{\tilde{x}_{h_t}=0},\\
		\label{eq:ap_4_4}
	\end{aligned}
\end{equation}
where $\tilde{x}_{h_t} = (\beta_{t+1}^*,\alpha_{t+1},\dots,\beta_{M+1}^*)^{\mathrm{T}}$ and $\tilde{z}_{h_t} = (\tilde{l}_{2t+1},\tilde{l}_{2t+2},\dots,\tilde{l}_{2M+1})^{\mathrm{T}}$.

As shown in Eq.~(\ref{eq:ap_4_2}) and~(\ref{eq:ap_4_4}), the analysis in Sec.~\ref{sec:3_2} is valid for any symmetric matrix $R$.

\begin{acknowledgements} 
	 We acknowledge the financial support in part by National Natural Science Foundation of China grant No.11974204 and No.12174215.
\end{acknowledgements}

\bibliographystyle{apsrev4-2}
\bibliography{ref}
\end{document}